\documentclass[11pt,a4paper]{article}

\usepackage{amsmath, amssymb, amsthm}
\usepackage{geometry}
\usepackage{hyperref}
\usepackage{microtype}
\usepackage{enumitem}
\usepackage{booktabs}
\usepackage{graphicx}
\usepackage{float}
\usepackage{enumitem}

\newtheorem{theorem}{Theorem}[section]

\newtheorem{definition}[theorem]{Definition}
\newtheorem{remark}[theorem]{Remark}

\title{Analytic patch trees: branch interface inheritance and fractal dimension fields}
\author{H.\ Mulder\\
\small\texttt{mulder@skystrategy.com}}
\date{2026}

\begin{document}

\maketitle

\begin{abstract}
The extension of the analytic fractal curve trees of~\cite{mulder2026}
to analytic surface patch trees reveals a new
geometric structure: branch points are replaced by interface curves
that transmit the full analytical state of parent patches to their
children. These interfaces prove to be central in 
determining the topology of the surface patch trees, including 
for the conditions for self-similarity of the interfaces, the patches
and thus the trees.

We establish the analytic conditions for the integrability and well-posedness
of the surface patch trees and introduce further restrictions for 
conformality. We demonstrate that patch trees have a natural foliation that
slices the trees into one dimensional curve trees, each of which has their
own Hausdorff dimension, jointly creating a smooth dimension field.

We extend the two dimensional surface model to arbitrary dimensions $n$ 
where $n-1$ interface manifolds transport the $n$ field state of the parent
patches to their child branches. We note that the balance or discrepancy between
patch field dimension and the dimensions in which the branches may evolve, 
determine the analytical regime from essentially geometrical to essentially
operational.
\end{abstract}

\bigskip


\medskip


\section{Introduction}
\label{sec:introduction}

In~\cite{mulder2026} we introduced analytic generator fields for fractal curve trees, showing 
that recursive fractal geometries can be generated through smooth field evolution 
rather than discrete geometric substitutions. The present work extends this framework 
from curves to surface patches. In doing so, branch points are upgraded to branch 
interfaces that transmit entire fields from parent to child patches. The resulting 
recursive geometry is governed not solely by the generator fields, but by the interaction 
between those fields and the evolving inherited interfaces. This leads naturally to the 
notion of an interface evolution operator, which becomes the organising principle for 
understanding and classifying patch trees.

A surface patch is specified by a pair $(V,\Gamma_0)$ consisting of generator 
fields $V=(V_1,V_2)$ and an analytic base interface $\Gamma_0$. When the generator fields 
satisfy the Frobenius integrability condition, they determine a unique analytic patch 
for any compatible inherited interface. 

This paper makes several contributions:
\begin{itemize}
    \item \textbf{Foundational structure:} We establish the analytic well-posedness of 
    surface patch trees via the Frobenius theorem and show that every such tree is 
    foliated by a smooth one-parameter family of fractal curve trees, each carrying 
    its own Hausdorff dimension. This yields a smooth dimension field across the patch 
    tree.
    \item \textbf{Interface theory:} We introduce branch interfaces as first-class 
    geometric objects. The interface evolution operator $E:\Gamma_0 \mapsto \Gamma_1$ 
    organises patch trees into geometric classes (interface preserving or modifying, 
    conformal, and self-similar) according to the behaviour of inherited interfaces under recursion.
    \item \textbf{Conformality and self-similarity:} Conformal generator fields satisfy 
    the Cauchy--Riemann equations and preserve angles. Within this family, constant-gradient 
    conformal fields generate the canonical self-similar conformal patch tree, in which 
    self-similarity emerges intrinsically from the interaction between the conformal 
    structure and interface evolution rather than through externally imposed transformations.
    \item \textbf{Higher-dimensional extension:} The patch tree framework generalises 
    naturally to arbitrary dimension. Patches of dimension $d_p$ are connected by 
    interfaces of dimension $d_p-1$ and extend the concepts of realisation, inheritance, 
    conformality, and self-similarity. The relative scaling between patch dimension and 
    branching (tree) dimension gives rise to qualitatively different regimes: geometric 
    versus operational.
\end{itemize}

The remainder of the paper is organised as follows. Section 2 introduces analytic surface 
patches, the generator system, Frobenius integrability, and the foliation by fractal 
curve trees. Section 3 develops the theory of branch interfaces and the interface evolution 
operator. Section 4 presents a hierarchy of patch tree classes of increasing geometric richness. 
Section 5 outlines the general higher-dimensional framework, and Section 6 concludes.

\section{Analytic surface patches and patch trees}
\label{sec:framework}

To extend fractal trees from curves to surfaces we use surface patches
and stitch them together along base and tip edges referred to as branch interfaces.
These interfaces are curve segments that are the equivalent of branch points 
in the curve tree framework. The surface patch tree, or patch tree for short, 
is entirely defined by the dynamics of generator fields of the patch as well 
as its base and tip interfaces.

\paragraph{The patch.}
Let $D = [0,W]\times[0,L]$ be the generator domain, with $s_1 \in [0,W]$
running across the patch and $s_2 \in [0,L]$ running into the patch interior
with base and tip at $s_2=0$ and $s_2=L$ respectively.
The generator state is $X: D \to \mathbb{R}^n$, whose components $x_n$ may
encode progression rates, orientation fields, curvature fields, as well as 
non-geometrical quantities. It evolves according to a pair of analytic
generator fields $V_1, V_2 : D \times \mathbb{R}^n \to \mathbb{R}^n$,
viewed as maps on the extended space $(s_1,s_2,X)$:
\begin{equation}
    \frac{\partial X}{\partial s_1} = V_1(s_1,s_2,X), \qquad
    \frac{\partial X}{\partial s_2} = V_2(s_1,s_2,X),
    \label{eq:generator}
\end{equation}
with analytic base interface $\Gamma_0 = \gamma(\cdot,0)$ specifying the bounday conditions. 
The pair $(V, \Gamma_0)$ fully specifies the patch and the constructed patch tree, assuming 
the following theorem is satisfied.

\begin{theorem}[Integrability and well-posedness]
\label{thm:integrability}
Let $V_1$ and $V_2$ be analytic. The system~\eqref{eq:generator}
admits a unique analytic solution on a neighbourhood of every point
of $D$ if and only if
\begin{equation}
    \frac{\partial V_1}{\partial s_2}
    + (V_2 \cdot \nabla_X)V_1
    =
    \frac{\partial V_2}{\partial s_1}
    + (V_1 \cdot \nabla_X)V_2.
    \label{eq:integrability}
\end{equation}
If the solution remains finite, it extends analytically to all of
$D$.
\end{theorem}

\begin{proof}
The solution requires path-independence of integration in
$(s_1,s_2)$. The closure condition for every coordinate rectangle
is precisely~\eqref{eq:integrability}, the Frobenius condition
for the associated lifted vector fields~\cite{lee2013}. Existence
and uniqueness of the analytic solution follow from the
Cauchy--Kowalevski theorem~\cite{john1982}.
\end{proof}

To move from state space to an embedded geometric realisation of a single patch, 
we perform the following projection.

\paragraph{Geometric realisation.}
Let $\Pi : \mathbb{R}^n \to \mathbb{R}^e$ be the projection of selected
components of $X$, including geometric coordinates and auxiliary quantities 
as needed. The \emph{realised patch} is the
map
\begin{equation}
    \gamma : D \to \mathbb{R}^e, \qquad
    \gamma(s_1,s_2) = \Pi\bigl(X(s_1,s_2)\bigr),
    \label{eq:realization}
\end{equation}
where $X$ is the unique analytic solution of~\eqref{eq:generator} and 
$\mathbb{R}^e$ is the embedding space. 

We are now in a position to assemble the patch tree.
\paragraph{The patch tree.}
A \emph{surface patch tree} is a recursively generated collection
of patches in which the tip interface of each recursive parent patch
branches into $N \geq 1$ child patches, each child inheriting the parent
state $X$ at the branch interface $\Gamma_L$ and their field generator $V$. The realised 
geometric structure is the embedded union
\begin{equation}
    \mathcal{T} = \bigcup_{P \in \mathcal{P}} \gamma_P(D_P)
    \;\subset\; \mathbb{R}^2,
    \label{eq:patch_union}
\end{equation}
where $\mathcal{P}$ is the index set of all patches and $D_P$ is
the generator domain of patch $P$. 

\begin{definition}[Analytic patch tree]
An \emph{analytic patch tree} is a patch tree in which

\begin{enumerate}
\item every patch is generated by an analytic generator system
      satisfying the Frobenius compatibility condition
      \eqref{eq:integrability};

\item every inherited interface is an analytic manifold;

\item every interface transport map
      \(T^{(k)} : \mathbb{R}^n \rightarrow \mathbb{R}^n\)
      is analytic; and

\item the inherited child state
      \[
      X^{(k)}_c(s_1,0)
      =
      T^{(k)}
      \bigl(
      X_p(s_1,L)
      \bigr)
      \]
      defines admissible boundary data for the child generator
      system, so that the Frobenius compatibility condition remains
      satisfied and the child patch admits a unique analytic
      realisation.
\end{enumerate}

Thus every patch in the tree possesses a unique analytic realisation,
and analyticity is preserved under recursive inheritance.
\end{definition}

The patch tree may be viewed as an assembly of curve trees, each generated from a point
on the base interface at the root of the tree. Each of these curve trees, 
in the contractive regime, tends to a well-defined fractal limit set, its attractor.

\begin{theorem}[Foliation by fractal curve trees]
\label{prop:foliation}
Let $A \in (0,1)$ denote the longitudinal contraction factor
accumulated across the patch depth $L$, and let $\rho_0(c)$ be the
progression rate at $s_1=c$ along the base interface
$\Gamma_0 = \gamma(\cdot,0)$. Assume that
\[
    r(c) := \rho_0(c)\,A^L
\]
satisfies $0 < r(c) < 1$ for every $c \in [0,W]$. Then:
\begin{enumerate}[label=(\roman*)]
\item The patch tree is foliated by a one-parameter family of
    smooth fractal curve trees: for each $c \in [0,W]$, the slice
    $s_1 = c$ is a smooth fractal curve tree in the sense
    of~\cite{mulder2026}, with a common branching topology.
\item Each slice carries a Hausdorff dimension
\begin{equation}
    d_{\mathrm{slice}}(c)
    =
    \frac{\log N}{\log(1/r(c))}.
    \label{eq:dimension_spectrum}
\end{equation}
\item The function
    $d_{\mathrm{slice}} : [0,W] \to \mathbb{R}_{>0}$
    is smooth. Thus the foliation carries a smooth
    \emph{slice-dimension field}.
\end{enumerate}
\end{theorem}

\begin{proof}
Fix $c\in[0,W]$. The slice $s_1=c$ inherits a one-dimensional
generator system with the same branching topology as in ~\cite{mulder2026} and
using Moran's  equation~\cite{falconer2003} yields the dimension 
formula~\eqref{eq:dimension_spectrum}. 
Smoothness of $d_{\mathrm{slice}}$ follows from the smoothness
of the interface progression-rate function $\rho_0(c)$ and the
assumption $r(c)\in(0,1)$.
\end{proof}

\begin{remark}[Local versus global Hausdorff dimension]
\label{rem:local_global_dim}
The quantity $d_{slice}(c)$ is the Hausdorff dimension of the curve tree
slice through $c$. The Hausdorff dimension of the full patch-tree attractor 
depends on the collective behaviour of the entire foliation and may not always
converge. Indeed, both the contraction rate along individual patches ($x_2$ direction) and
contraction of the inherited interfaces ($x_1$ direction) must be considered.
The relationship between the dimension field and the global attractor dimension 
is an open problem.
\end{remark}

\section{Branch Interfaces}
\label{sec:interfaces}

The interface between parent and child patches is more than a geometrical
boundary. It carries the parent state $X$, including position and tangent fields
and any other state variables as defined, to be inherited by the child branches. 
It also acts as the source of child multiplicity, i.e. the number of branches 
initiated by the parent. The interface may carry a transformation function 
that modifies the inherited state that may introduce tangent and position offsets. 
The latter is not treated in the paper because we restrict ourselves to 
analytic patch trees. 
\subsection{Interface inheritance}
\label{sec:inheritance}

\begin{definition}[Branch interface and inheritance]
\label{def:interface}
The \emph{branch interface} is the tip curve
\[
    \Gamma(s_1) = \gamma(s_1, L), \qquad s_1 \in [0,W],
\]
carrying the full field of generator state $X(s_1,L)$ from parent
to child. Each child $k$ inherits the parent tip state to be their own 
base interface via an analytic \emph{modification} 
$T^{(k)}:\mathbb{R}^n\to\mathbb{R}^n$:
\begin{equation}
    X_{\mathrm{c}}^{(k)}(s_1,0)
    = T^{(k)}\!\bigl(X_{\mathrm{p}}(s_1,L)\bigr),
    \qquad s_1\in[0,W].
    \label{eq:interface_def}
\end{equation}
Transport may be \emph{transparent}: $T^{(k)}=\mathrm{id}$. Child patches
become mutually distinguishable only when a transport modification is
applied. Transport modifications do not affect the integrability of the
generator system, but for analytic patch trees they must preserve 
$C^k$-compatibility between the parent and child patches along the inherited interface.
\end{definition}

\begin{remark}[Junction geometry]
\label{rem:junction}
Apparent overlaps in the planar embedding of branching children
are interpreted as projection artefacts and do not affect the
recursive construction when each child patch is assumed to occupy
a separate geometric layer attached along the common branch
interface. In applied settings, the overlap has consequences and 
this assumption needs to be revisited.
\end{remark}

\subsection{Interface evolution}
\label{sec:evolution}

We observe that the evolution of the interfaces through the branch generations
is itself a defining feature of the patch tree.

\begin{definition}[Interface evolution operator]
\label{def:evolution}
Given a patch specification $(V,\Gamma_0)$, define the
\emph{interface evolution operator}
\begin{equation}
    E : \Gamma_0 \mapsto \Gamma_1,
    \label{eq:evolution}
\end{equation}
where
\[
    \Gamma_1(s_1)
    =
    \gamma(s_1,L)
\]
is the tip interface obtained by integrating the generator system
from the base interface $\Gamma_0$.
\end{definition}

Together with the seam transport maps $T^{(k)}$, the interface evolution operator $E$
governs the recursive geometry of the entire patch tree.

\section{Patch tree classes}
\label{sec:examples}

In this section we explore a few types of patch trees that illustrate the setup 
and construction of patch trees and points at an elementary classification.

\subsection{Interface-preserving evolution}
\label{sec:uniform}

The simplest class of patch trees fully preserves the interface across generations.
The generator fields are assumed to be independent of the interface
coordinate $s_1$:
\[
    \rho(s_1,s_2) = \rho(s_2), \qquad
    \theta(s_1,s_2) = \theta(s_2).
\]

Thus every foliation experiences exactly the same progression and rotation 
as it propagates through the patch. 

Let the base interface be the straight segment
\[
    \Gamma_0(s_1) = p_0 + s_1u,
\]
where $u \in \mathbb{R}^2$ is a unit vector. Using the basic curve tree example from 
~\cite{mulder2026} the realisation is

\[
    \gamma(s_1,s_2)
    =
    \Gamma_0(s_1)
    +
    \int_0^{s_2}
    \rho(t)
    \begin{pmatrix}
        \cos\theta(t)\\
        \sin\theta(t)
    \end{pmatrix}
    \, dt.
\]

The inherited tip interface is

\[
    \Gamma_1(s_1)
    =
    \Gamma_0(s_1)
    +
    \int_0^{L}
    \rho(t)
    \begin{pmatrix}
        \cos\theta(t)\\
        \sin\theta(t)
    \end{pmatrix}
    \, dt.
\]

Thus the interface evolution operator acts by translation:
\[
    E(\Gamma_0)
    =
    \Gamma_0
    +
    L\rho_0
    \begin{pmatrix}
        \cos\theta_0\\
        \sin\theta_0
    \end{pmatrix}.
\]
The inherited interface has exactly the same geometric form as the
base interface. The interface class is therefore preserved under
evolution.

\begin{figure}[htbp]
    \centering
    \includegraphics[width=0.7\textwidth]{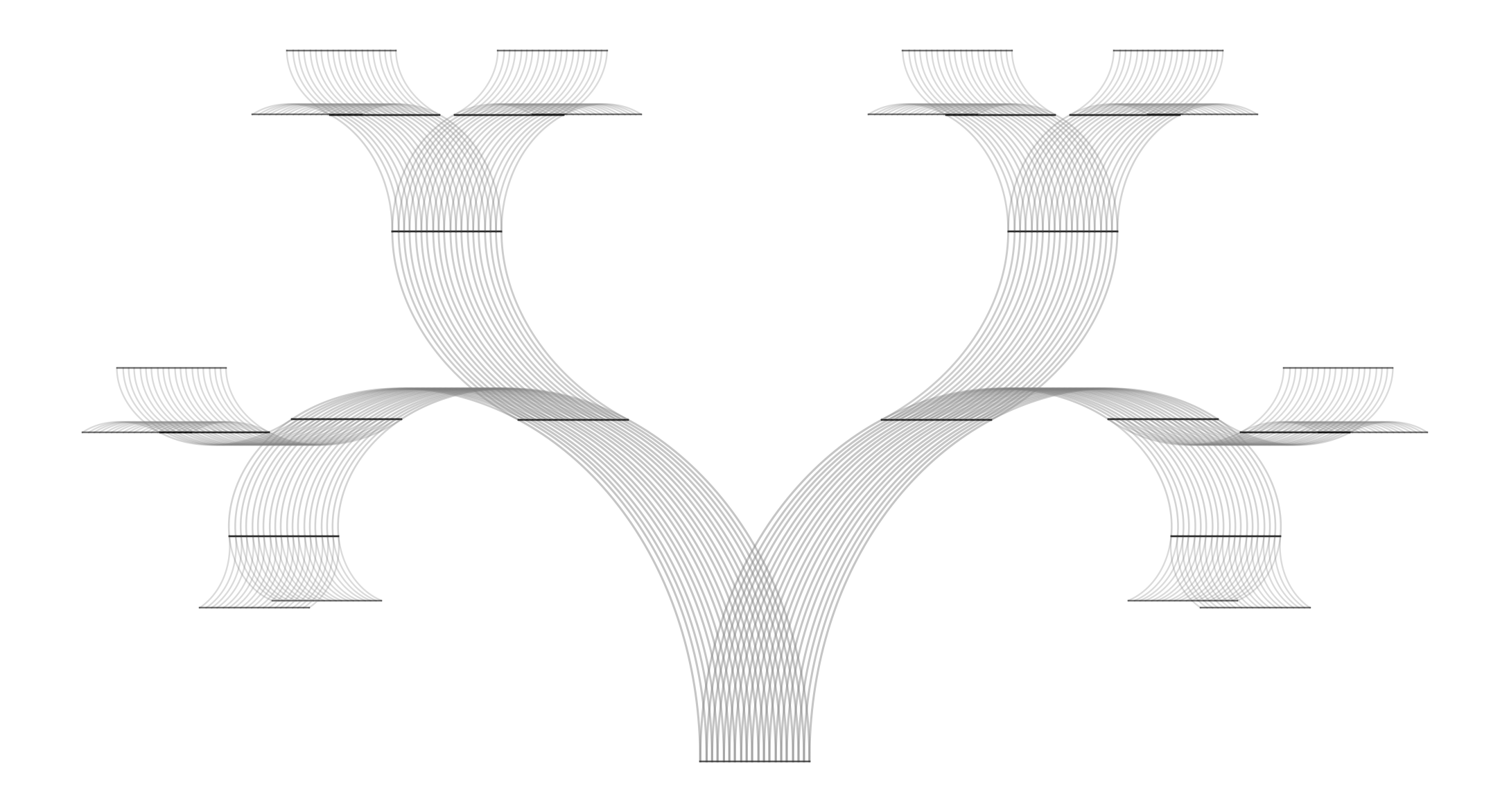}
    \caption{
        Uniform patch tree generated from generator fields
        $\rho=\rho(s_2)$ and $\theta=\theta(s_2)$, independent
        of the interface coordinate $s_1$, resulting in parallel foliations.
        The generator fields preserve the geometry of inherited
        interfaces (black horitontal lines), so that the interface evolution operator acts only
        by translation.}
    \label{fig:uniform_tree}
\end{figure}

\subsection{Interface-modifying evolution}
\label{sec:heterogeneous}

When the generator fields vary spatially, they deform inherited
interfaces as those interfaces propagate through a patch generator field.

As an example, let
\begin{equation}
\rho(s_1,s_2)
=
\rho_b(s_1)\,A^{s_2},
\qquad
\rho_b(s_1)
=
\rho_{\min}
+
(\rho_{\max}-\rho_{\min})
\sin^2\!\left(\frac{\pi s_1}{W}\right),
\label{eq:rho_nonuniform}
\end{equation}
and
\begin{equation}
\theta(s_1,s_2)
=
\theta_0+\omega s_2.
\label{eq:theta_nonuniform}
\end{equation}

where $0<\rho_{\min}<\rho_{\max}$ define a
non-uniform progression profile across the interface. $A\in(0,1)$
governs longitudinal contraction, and $\omega$ determines the rate of
rotation along the progression direction. The resulting realisation therefore 
deforms the inherited interface geometry.

The geometric realisation along the patch satisfies
\begin{equation}
\frac{\partial \gamma}{\partial s_2}
=
\rho(s_1,s_2)
\begin{pmatrix}
\cos\theta(s_1,s_2)\\
\sin\theta(s_1,s_2)
\end{pmatrix}.
\label{eq:realisation_nonuniform}
\end{equation}

Integrating from the base interface $s_2=0$ to the tip interface
$s_2=L$ gives
\begin{align}
\Gamma_1(s_1)
&=
\Gamma_0(s_1)
+
\int_0^L
\rho(s_1,s_2)
\begin{pmatrix}
\cos\theta(s_1,s_2)\\
\sin\theta(s_1,s_2)
\end{pmatrix}
ds_2
\\
&=
\Gamma_0(s_1)
+
\rho_b(s_1)
\int_0^L
A^{s_2}
\begin{pmatrix}
\cos(\theta_0+\omega s_2)\\
\sin(\theta_0+\omega s_2)
\end{pmatrix}
ds_2.
\label{eq:tip_interface_nonuniform}
\end{align}

We define the displacement vector
\begin{equation}
v
=
\int_0^L
A^{s_2}
\begin{pmatrix}
\cos(\theta_0+\omega s_2)\\
\sin(\theta_0+\omega s_2)
\end{pmatrix}
ds_2,
\label{eq:displacement_vector}
\end{equation}

then find the interface evolution operator
\begin{equation}
E(s_1)
=
\Gamma_0(s_1)
+
\rho_b(s_1)\,v.
\label{eq:interface_evolution_nonuniform}
\end{equation}.

The $\rho_b(s_1)$ perturbation field itself repeats accumatively from
generation to generation. 

\begin{figure}[htbp]
    \centering
    \includegraphics[width=0.7\textwidth]{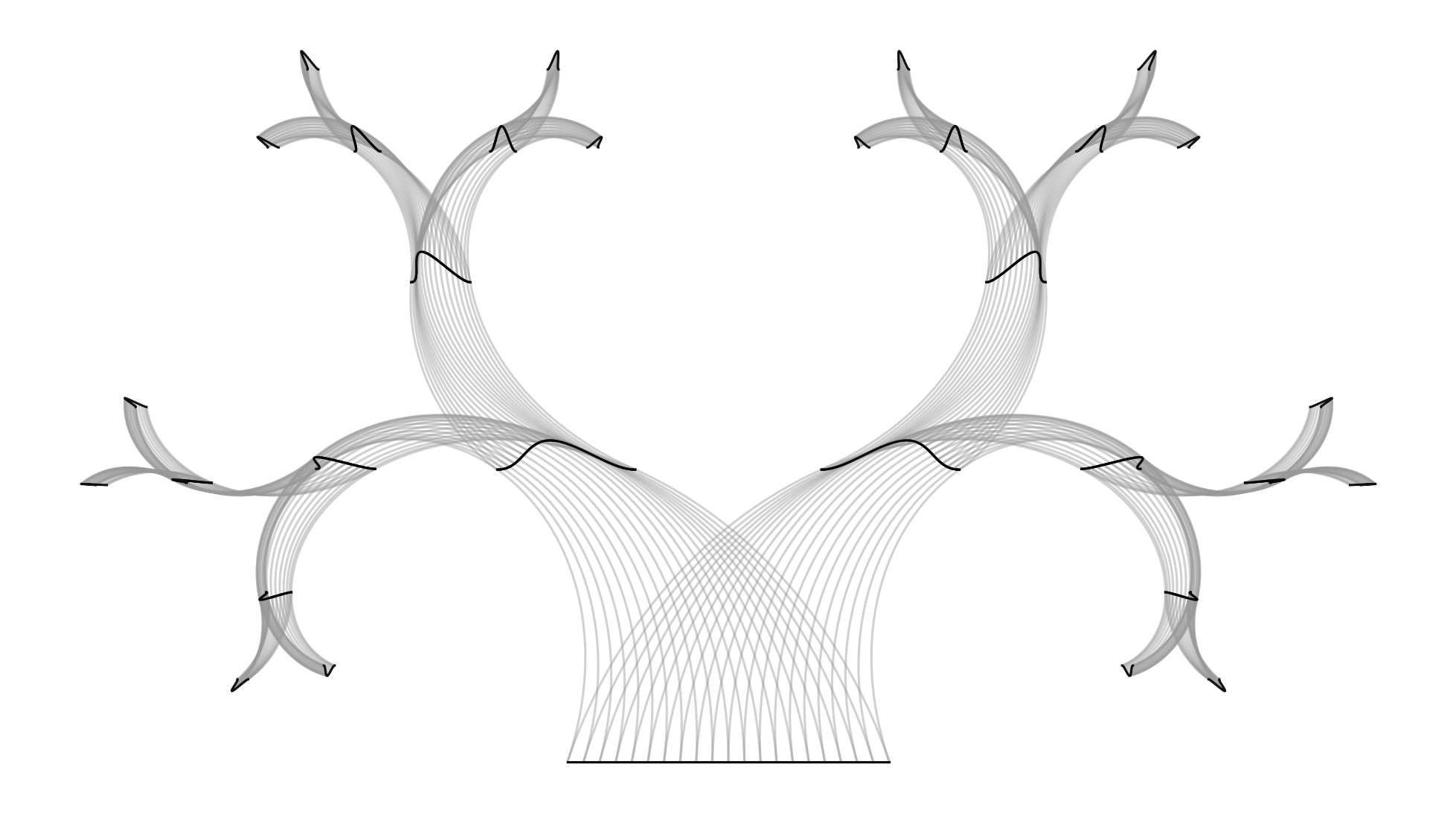}
    \caption{Patch tree with with an interface evolution operator. 
    Tip interfaces (black curves) are geometrically distinct from their base
    interfaces, increasingly modified by the $\rho_b(s_1)$ pertubation field 
    with each generation.}
    \label{fig:nondegenerate_patch_tree}
\end{figure}

\subsection{Conformal patch trees}
\label{sec:conformal_general}
We introduce conformal patch trees.
\begin{theorem}[Conformal Patch Tree Theorem]
\label{thm:conformal_patch_tree}

Let $(\rho,\theta)$ be analytic generator fields satisfying

\begin{equation}
    \frac{\partial}{\partial s_1}\log\rho
    = \frac{\partial\theta}{\partial s_2},
    \qquad
    \frac{\partial}{\partial s_2}\log\rho
    = -\frac{\partial\theta}{\partial s_1}.
    \label{eq:CR_equations}
\end{equation}

on the generator domain $D$.

Then the realised patch map $\gamma:D\to\mathbb{R}^2$
is locally conformal wherever $\rho>0$.
The induced metric takes the form

\[
g=\lambda^2 I,
\]

for some positive scalar field $\lambda$, and the
foliation directions remain orthogonal throughout the
patch realisation.

Equivalently, the generator field defines a complex-analytic
structure on the generator domain, and the realised patch
preserves angles locally.
\end{theorem}

\begin{proof}
Define the complex field

\[
\phi(s_1,s_2)
=
\log\rho(s_1,s_2)
+
i\theta(s_1,s_2).
\]

The conditions \eqref{eq:CR_equations} are precisely the
Cauchy--Riemann equations for $\phi$, implying that
$\phi$ is holomorphic.

The realisation derivative may therefore be written in
complex form as

\[
F'(z)=e^{\phi(z)},
\qquad z=s_1+i s_2,
\]

which is holomorphic and non-vanishing wherever
$\rho>0$.

A holomorphic map with non-zero derivative is locally conformal
\cite{ahlfors1979complex}. Since \(F\) is locally conformal, the induced metric on
the generator domain is

\[
g
=
|F'(z)|^2 I
=
\rho^2 I.
\]

Thus

\[
g=\lambda^2 I,
\qquad
\lambda=\rho,
\]

showing that the realised patch differs from the Euclidean
plane only by a local isotropic scaling. Consequently the
coordinate foliations remain orthogonal and angles are
preserved.

Orthogonality of the coordinate foliations follows
immediately from the diagonal form of the metric.
\end{proof}
The importance of the conformal subclass extends beyond its geometric
properties. By placing the generator fields within a complex-analytic
framework, it facilitates access to harmonic and spectral techniques that 
motivates the research into these analytical patch trees. 

\subsubsection{Canonical self-similar conformal trees}
\label{sec:conformal}

We now introduce a subclass of conformal patch trees for which the inherited 
interfaces evolve by similarity transformations.

\begin{theorem}[Canonical Self-Similar Conformal Tree]
\label{cor:canonical_selfsimilar}

Let the conformal generator field be

\[
    \phi(z)=A+kz,
\]

with $A,k\in\mathbb{C}$ constant. Then the corresponding realised
patch map satisfies

\[
    F(z+iL)
    =
    aF(z)+b,
\]

where

\[
    a=e^{ikL},
    \qquad
    b=(1-a)F_0.
\]

Consequently, the interface evolution operator acts by a similarity
transformation on the inherited interface family. The resulting
conformal patch tree is therefore geometrically self-similar, with
successive interfaces belonging to a single similarity class.

In particular, writing

\[
k=\alpha+i\beta,
\]

the similarity is contractive whenever

\[
\beta>0.
\]

\[
    F_*=\frac{b}{1-a}=F_0.
\]

\end{theorem}
\begin{proof}
Since $\phi(z)=A+kz$, we have

\[
    F'(z)=e^{\phi(z)}=e^A e^{kz}=Ce^{kz}.
\]

Integrating gives

\[
    F(z)=De^{kz}+F_0,
\]

for constants $D,F_0\in\mathbb{C}$. Therefore

\[
    F(z+iL)-F_0
    =
    De^{k(z+iL)}
    =
    e^{ikL}De^{kz}
    =
    e^{ikL}\bigl(F(z)-F_0\bigr).
\]

Hence

\[
    F(z+iL)
    =
    e^{ikL}F(z)+(1-e^{ikL})F_0.
\]

Setting

\[
    a=e^{ikL},
    \qquad
    b=(1-a)F_0,
\]

we obtain

\[
    F(z+iL)=aF(z)+b.
\]

This is a similarity transformation of the realised plane whenever
$a\neq 0$, which holds because $a=e^{ikL}$. Thus the tip interface is
similar to the base interface, and the interface evolution operator
preserves the similarity class of the inherited interfaces. Repeating
the same argument generation by generation shows that all successive
interfaces are related by iterates of the same similarity map.

If $|a|<1$, these iterates converge to the fixed point $F_*$ satisfying

\[
    F_*=aF_*+b.
\]

Solving gives

\[
    F_*=\frac{b}{1-a}=F_0.
\]

Since

\[
    |a|=|e^{ikL}|,
\]

the contraction condition is determined by the imaginary part of $k$.
For $k=\alpha+i\beta$,

\[
    ikL=i\alpha L-\beta L,
\]

and therefore

\[
    |a|=e^{-\beta L}.
\]

Thus the similarity is contractive precisely when $\beta>0$.
\end{proof}

The significance of this result is that self-similarity is not imposed
externally through branch transformations. Instead, it emerges from
the interaction between the conformal field system and the inherited
interface geometry.

\begin{remark}[Base interface requirement for conformal patch trees]
In the canonical conformal construction the base interface is not
freely specifiable. It is the image of the boundary

\[
s_2=0
\]

under the holomorphic realisation map

\[
F(z)=D e^{kz}+F_0.
\]

Consequently the base interface is

\[
\Gamma_0(s_1)
=
F(s_1)
=
F_0 + D e^{k s_1}.
\]

For

\[
k=i\beta,
\]

the interface is a circular arc centred at \(F_0\). More generally,
for

\[
k=\alpha+i\beta,
\]

the interface is a logarithmic spiral,

\[
\Gamma_0(s_1)
=
F_0 + D e^{\alpha s_1} e^{i\beta s_1}.
\]

The inherited interfaces are therefore not arbitrary, but form a
nested similarity family generated by the same conformal map. In
particular,

\[
\Gamma_n(s_1)
=
F_0 + a^n\bigl(\Gamma_0(s_1)-F_0\bigr),
\qquad
a=e^{ikL},
\]

so that all inherited interfaces belong to a single similarity class
and converge toward the fixed point \(F_0\) whenever \(|a|<1\).
\end{remark}

\begin{figure}[htbp]
    \centering
    \includegraphics[width=0.7\textwidth]{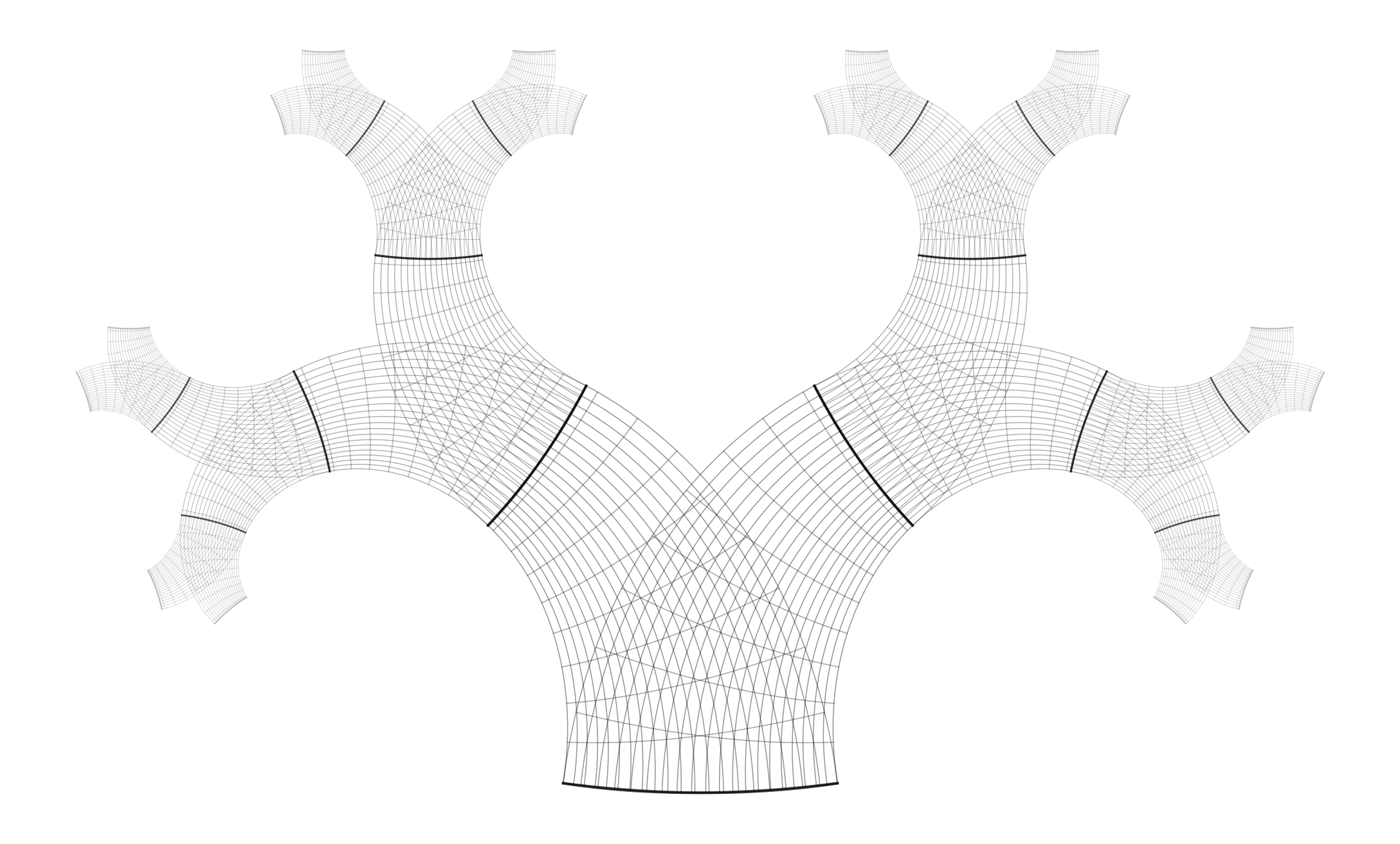}
    \caption{Canonical self-similar conformal patch tree with
    foliations that remain orthogonal throughout the realisation.
    Successive inherited interfaces are circular arcs whose radii
    decrease by the factor $e^{-\beta L}$ at each generation.}
    \label{fig:ag2_tree_uniform}
\end{figure}

\subsection{Elementary classification of patch trees}
\label{sec:classification}

The preceding examples suggest that the behaviour of the interface
evolution operator

\[
E:\Gamma_0 \mapsto \Gamma_1
\]

provides a natural basis for classifying patch trees.

The simplest class consists of interface-preserving trees, for which
\(E\) acts by translation and inherited interfaces retain their
geometric form. More generally, interface-modifying trees arise when
the generator fields deform the inherited interfaces as they propagate
through a patch. Within this broader class, conformal patch trees are
distinguished by the Cauchy--Riemann conditions on the generator
fields, giving rise to locally angle-preserving geometries. A further
subclass is formed by the canonical self-similar conformal trees, in
which the interface evolution operator acts by a similarity
transformation and successive inherited interfaces belong to a single
similarity class.

This classification highlights a central feature of the framework:
the recursive geometry is determined not solely by the generator
fields nor solely by the inherited interfaces, but by their coupled
evolution. The interface evolution operator therefore plays a role
analogous to the recursion operator in classical self-similar
constructions, providing a natural organising principle for patch
tree geometry.

\section{Towards higher-dimensional patch trees}
\label{sec:agn}

The patch tree framework is not limited to two dimensions. The
central concepts such as generator fields, Frobenius integrability,
interface inheritance, conformality, and geometric
self-similarity extend to arbitrary dimension. 

\subsection{Higher-dimensional patch realisation}
\label{sec:agn_realization}

Let $D = [0,W_1]\times\cdots\times[0,W_n]$ be an $n$-dimensional
generator domain. A patch is described by a generator state
$X : D \to \mathbb{R}^m$ satisfying
\begin{equation}
    \frac{\partial X}{\partial s_i} = V_i(s,X),
    \qquad i = 1,\ldots,n,
    \label{eq:agn_generator}
\end{equation}
where the fields $V_i$ are analytic and satisfy the Frobenius
compatibility conditions
\begin{equation}
    [V_i, V_j] = 0, \qquad 1 \leq i,j \leq n.
    \label{eq:agn_frobenius}
\end{equation}
The base and tip interfaces are $(n-1)$-dimensional manifolds and like the 
two dimensional case:
\begin{itemize}
\item carry the parent's higher dimensional field state to child branches
\item interfaces evolve via the interface evolution operator 
\item the assembled patch trees are foliations of slices across the underlying 
dimensions and 
\item the slice Hausdorff dimension field has itself dimension $(n-1)$.
\end{itemize}
\subsection{Patch dimension vs branching dimension}
\label{sec:patch_tree_dim}
The patch dimension, interface dimension, embedding dimension, and branching dimension are
independent characteristics of a recursive manifold. Different analytical regimes arise 
according to their relative magnitudes.

We observed that the choice of patch vs branching dimension is an indication 
of the problem space we are analyzing.

\paragraph{Combined regime: $d_p \approx d_b$.}
Neither recursive organisation nor patch geometry dominates: this is the domain of
recursive geometry, interface inheritance, conformality, attractor structure, and analysis 
on recursive spaces. Classical analysis on fractals naturally belongs to this regime.

\paragraph{Geometric regime: $d_p \ll d_b$.}
The recursive organisation dominates: the generator fields of individual patches serves 
primarily to drive the interface evolution. Questions focus on recursive 
topology, symbolic structure, attractors, self-similarity, and invariants of the 
branching dynamics.

\paragraph{Operational regime: $d_p \gg d_b$.}
Patch geometry dominates: the recursive structure acts as a distribution network, 
connecting complex local patch objects and the areas of interest are transport, spectral 
theory, diffusion, wave propagation, and operator dynamics.

\section{Conclusion}
\label{sec:conclusion}

This paper took the one dimensional analytic curve trees as a starting point and extended
this to surface patch trees and then to trees manifolds of arbitrary dimensions. The analyticity of these
patch trees are the foundation for everything that followed. In transitioning from curves 
to surfaces, we found that the mundane branch points in curve trees unfold into interface 
curves that mediate geometric and non-geometric data from parent surface patches to their 
children. Interface curves proved to be more than a surface patch boundary but are 
central geometric objects in their own right. They drive the topology of the patch tree.

We established the conditions for integrability and well-posedness of the generator fields 
and added conditions for conformal surface patches. Several new, and somewhat surprising 
objects emerged as we unpacked the patch trees. We found that patch trees, of any dimension, 
can be expressed as a foliation of one dimensional curve trees. Each of these foliations 
carry their own Hausdorff dimension an across the full foliation establish a smooth dimension 
field.

In exploring several examples, it became evident that the interplay between patch generators and 
their base and tip interfaces naturally produce an elementary classification. We determined the 
conditions for self-similar tree topologies, which is purely an interface dependency and the patch
generator on the other hand, solely determines whether patches are conformal.

This work indicates several directions to be explored. Beyond the self-similar and conformal
characterization, it is likely that other geometric regimes exist, such as periodic and 
quasi-periodic. We focused on symmetric interfaces but asymmetric interfaces will lead to another 
class of patch trees. Although we defined a transport modification on the parent states, 
including the introduction of discontinuities and generation dependency, we have yet to explore
that feature.

Another rich vein is the study of spectral and analytical
properties of patch trees. Classical fractal analysis seeks to define
Laplacians and related operators directly on geometrically irregular
limit sets. By contrast, conformal patch trees are assembled from
smooth recursive manifolds on which the full machinery of differential
geometry and partial differential equations is already available. The
fractal geometry emerges only as a recursive limit of these smooth
structures.

This raises the possibility that difficult analytical questions on
certain classes of fractals may be approached indirectly, and perhaps 
paradoxically, through their smooth recursive approximations. Whether 
such an approach can recover known spectral results or yield new ones 
remains an open question, but the framework provides a natural setting 
in which these questions may be explored.

\end{document}